\documentclass[5p]{elsarticle}
\journal{Computers \& Security}
\usepackage[utf8]{inputenc}
\usepackage{graphicx}
\usepackage[figurename=Figure]{caption}
\usepackage{enumitem}
\PassOptionsToPackage{hyphens}{url}\usepackage{hyperref}
\usepackage{color, colortbl}
\definecolor{Gray}{gray}{0.9}
\usepackage{array}
\usepackage{soul,cancel}
\usepackage{float}
\usepackage{stfloats}
\usepackage{multirow}

\newcolumntype{P}[1]{>{\centering\arraybackslash}p{#1}}
\usepackage{mathtools}

\usepackage{algorithm}
\usepackage{algorithmicx}
\usepackage{algcompatible}
\usepackage[end]{algpseudocode}

\usepackage{amsthm,soul,comment}
\theoremstyle{definition}
\newtheorem{definition}{Definition}
\newtheorem{theorem}{Theorem}

\newcommand{\newHL}{\text}
\begin{document}\sloppy
\begin{frontmatter}
	\title{Mahalanobis distance-based robust approaches against false data injection attacks on dynamic power state estimation}
	
	\author[First]{Jing Lin}
	\ead{jinglin@usf.edu}
	
	\author[First]{Kaiqi Xiong \corref{cor1}}
	\ead{xiongk@usf.edu}
	
	\address[First]{ICNS Lab and Cyber Florida, University of South Florida, Tampa, FL33620, USA}
	\cortext[cor1]{Corresponding author.}

\begin{abstract}
Many researchers have studied false data injection (FDI) attacks in power state estimation, but existing state estimation approaches are still highly vulnerable to FDI attacks. In this paper, we investigate the problem of the above three FDI attacks against dynamic power state estimation (DSE). Although the three attacks were discovered in SSE several years ago, none of them has been well addressed in static power state systems. In this research, we propose two robust defense approaches against the above three efficient FDI attacks on DSE. Compared to existing approaches, our proposed approaches have three major differences and significant strengths: (1) they defend against the three FDI attacks on dynamic power state estimation rather than static power state estimation, (2) they give a robust estimator that can accurately extract a subset of attack-free sensors for power state estimation, and (3) they adopt the little-known Mahalanobis distance in the consistency check of power sensor measurements, which is different from the Euclidean distance used in all the existing studies on power state estimation. We mathematically prove that the Mahalanobis distance is not only useful but also much better than the Euclidean distance in the consistency check of power sensor measurements. Our time complexity analysis shows that the two proposed robust defense approaches are efficient. Moreover, in order to demonstrate the effectiveness of the proposed approaches, we compare them with the three well-known approaches: the least square approach, the Imhotep-SMT approach, and the MEE-UKF approach. Our extensive experiments show that the proposed approaches further reduce the estimation error by two orders of magnitude and four orders of magnitude compared to the Imhotep-SMT approach and the least square approach, respectively. Moreover, our approach is more stable than the MEE-UKF approach.  
\end{abstract}

\end{frontmatter}
	
\section{Introduction}~\label{introduction}
The power grid is a complex electricity delivery system of interconnected networks that deliver electricity from electric power generators to large geographical areas through transmission and distribution lines. Our daily routines rely heavily on the use of electricity. However, a few recent incidents have indicated that the power system is under constant attacks~\cite{cyber}. In December 2015, hackers struck three power distribution centers in Western Ukraine and caused nearly a quarter-million residents to lose their power for several hours~\cite{1}. A year later, the Ukrainian power grid was attacked again and people of the city of Kiev lost power for about an hour~\cite{u3, 101}. Such incidents are not limited to Ukraine. Recently, the Department of Homeland Security and the Federal Bureau of Investigation (FBI) alerted the public that the foreign government has attempted to target critical US infrastructure sectors since at least March 2016 \cite{b5}. Stuxnet malware is the first well-known cyberwarfare weapon that targeted Iran's nuclear power program~\cite{Stuxnet}, \cite{bs}. Two consecutive Ukrainian power grid attacks could be a test run for another cyberwarfare weapon. It is critical to protect our infrastructure from cybercriminals. 

The real-time system monitoring of a power system is essential to ensure the reliable and secured electricity service operation of power grids. The control center of the power system uses the collected sensor measurements to conduct state estimation (SE). SE is then used by the power system operators to perform the contingency analysis for power system security.  

There are two types of SE: traditional static state estimation (SSE) and dynamic state estimation (DSE)~\cite{wu, ding2018survey}. SSE does not consider the relationship among the states at a different time (i.e., states are not varied with time), whereas DSE does. SSE usually uses some sort of the least square approaches such as the weighted least square approach to obtain the best estimate of the static state variables. On the other hand, Kalman filter techniques are widely used for real-time DSE to obtain the optimal estimate of the power grid state~\cite{8, huang2007feasibility, 10}. Both types of SE use power flow models and meter measurements. A power flow model consists of a set of power flow equations that describe the flow of electric power in an interconnected system. The SE based on these nonlinear equations can be computationally expensive. Instead, the linearized power flow model is often used to approximate the power flow model, although it is less accurate. 

Cybercriminals have many ways to attack a power system. One of the most severe ways to attack the power system is through FDI. For instance, an attacker can compromise a few sensors and inject malicious data to mislead SE and intervene in normal power system operations, such as the attacks occurred in Ukraine \cite{1}. In this paper, we investigate FDI attacks against {\it dynamic} power state estimation.

To address this threat, researchers have come up with various defense methods, but they have only focused on {\it static} power state estimation. The Imhotep-SMT approach~\cite{smt} and the least square approach~\cite{3} are two popular ways to solve the SE problem.
The Imhotep-SMT approach uses the Satisfiability Modulo Theories to solve the state estimation problem, whereas
the least square approach estimates the state of a system by
minimizing the mean square error. Furthermore, in~\cite{lin}, a consistency check among sensor readings is introduced for SSE under the FDI attacks. The formal definition of such a consistency check is presented in section II. However, it is shown in \cite{b10} that there are several attack approaches that can inject false data and bypass the consistency check. In \cite{b12}, the Kalman filter approach uses the information from a system model, previous estimation of system states, and sensor measurements to give the optimal SE. Nevertheless, if the sensor measurement is modified by an attacker, the Kalman filter approach provides misleading state estimation since this approach gives the same weight to these malicious sensors as to the normal sensors when performing the state estimation. 

To find a robust method for SE against FDI attacks, we combine the consistency check with the Kalman filter approach to defending against it. First, we present the probabilistic rank-based expanding approach to finding a large set of consistent sensors for a dynamic power system and then apply the Kalman filter approach to this subset of the consistent sensors for SE. It is called the combined consistency and Kalman filter approach (CCKF). To this end, we review the specific attacks that target the consistency check. Then, we introduce our attack-resilient approach and evaluate it by conducting extensive experiments on an IEEE 14-bus system. Furthermore, we perform the runtime analysis to show the scalability of our approach. 

The key contributions of this paper are in the following:
\begin{itemize}
    \item All existing studies on SE use the  Euclidean distance that does not consider the relationship among sensor measurements. Instead, we adopt the Mahalanobis transformation and the Mahalanobis distance to address the correlation between the measurement noise and standardize the variance of each sensor's measurements. Specifically, we mathematically prove that the Mahalanobis distance is a better measure of the error than the Euclidean distance does, so it is used for consistency check.
    
    \item To improve the robustness of SE, we propose two new approaches for selecting a consistent set of sensor measurements for dynamic power state estimation. Our new proposed approaches do not rely on solving $\textbf{Cx}=\textbf{y}$ to find the MMSE estimate for $\textbf{x}$ and so Liu et al.'s method~\cite{b10} for generating false data attack vectors based on the column space of $\textbf{C}$ does not work. That is, our proposed approaches are resilient to the false data attacks discovered in~\cite{b10}. They are also robust statistics since they provide a good performance no matter what distributions sensor measurement data follow when the number of the sensors is large based on the central limit theorem.
    
    \item We propose the prediction-based consistency approach (PCNA) and the combined consistency check with the Kalman filter (CCKF) approach for dynamic power state estimation against FDI attacks. The two proposed approaches, which incorporate a consistency check into the Kalman filter for the state estimation, have similar performance in terms of accuracy, though PCNA is more time-efficient. Furthermore, the experimental results show that the two proposed approaches outperform the two well-known approaches: the Imhotep-SMT approach~\cite{smt} and the least square approach. Specifically, the two proposed approaches reduce the estimation error by two orders of magnitude and four orders of magnitude compared to the Imhotep-SMT approach and the least square approach, respectively.

\end{itemize}

The remainder of this paper is organized as follows. In section~\ref{model}, we give the background information of this research, such as the system model. In addition, section~\ref{model} introduces the concept of $\tau$-consistency and discusses attack models. Section~\ref{method} first studies the Mahalanobis transformation for obtaining uncorrelated and standardize variables, as well as the Mahalanobis distance for measuring distance when the variables are correlated and have the fluctuations of different magnitudes. Then, section~\ref{method} presents PCNA and CCKF approaches for SE. In section~\ref{method}, we further conduct the time complexity analysis of these approaches. Section~\ref{evaluation} illustrates the performance of these approaches against FDI attacks. Section~\ref{related} discusses related work. Finally, we conclude this paper and point out some future research directions in section~\ref{conclusion}.

\section{System Model and Problem Formulations} \label{model}
In this section, we start with necessary mathematical notation, present the system model, define measurement consistency with existing approaches against false data injection, and give three attack models.

\subsection{Notation}
In this paper, we denote a vector by a boldfaced lowercase letter, and a matrix by an boldfaced uppercase letter. The symbol $\mathbf{R}$ denotes a set of real numbers, as well as $k$, $n$, and $p$ are positive integers. The symbol $1_n$ denotes a $n$ by 1 vector whose elements are all 1's, $0_n$ denotes a $n$ by 1 vector whose elements are all 0's, and $\textbf{I}_n$ denotes the $n$ by $n$ identity matrix. $\textbf{C}(j_1,...,j_k)$ is a matrix consisting of the $j_1$-th, $j_2$-th,..., and $j_k$-th rows of matrix $\textbf{C}$. The symbol $\textbf{A}^{-}$ (or $\textbf{A}^{-1}$) represents the generalized inverse (or inverse if $\textbf{A}$ is invertible) of matrix $\textbf{A}$. Note given a matrix $\textbf{A}$, $\textbf{A}^{-}$ is a generalized inverse of $\textbf{A}$ if it satisfies the condition $\textbf{AA}^-\textbf{A}=\textbf{A}$ The transpose of a matrix $\textbf{A}$ is denoted by $\textbf{A}^T$. The symbol $\otimes$ denotes a Kronecker product, a generalization of the outer product. Given a $d \times n$ matrix through 
\[
\textbf{A}=
\begin{bmatrix}
    a_{11}       & a_{12} & a_{13} & \dots & a_{1n} \\
    a_{21}       & a_{22} & a_{23} & \dots & a_{2n} \\
    \hdotsfor{5} \\
    a_{d1}       & a_{d2} & a_{d3} & \dots & a_{dn}
\end{bmatrix}
\]
 and a $k \times p$ matrix $\textbf{B}$, then the Kronecker product:
\[
\textbf{A} \otimes \textbf{B}
=
\begin{bmatrix}
    a_{11} \textbf{B}      & a_{12}\textbf{B} & a_{13}\textbf{B} & \dots & a_{1n}\textbf{B} \\
    a_{21} \textbf{B}      & a_{22}\textbf{B} & a_{23}\textbf{B} & \dots & a_{2n}\textbf{B} \\
    \hdotsfor{5} \\
    a_{d1}\textbf{B}       & a_{d2}\textbf{B} & a_{d3}\textbf{B} & \dots & a_{dn}\textbf{B}
\end{bmatrix}
\]
is a $dk \times np$ matrix. The $L_2$-norm of a $p \times 1$ vector $x$ is defined by:
 $$||  \textbf{x}||=\sqrt{\sum_{i=1}^p|x_i|^2},$$ 
where $|x_i|$ is the absolute value of $x_i$. 
\subsection{System Model}

Assume that there are $p$ power state variables to be estimated in a power system, each variable representing either the voltage magnitude $V_i$ or phase angle $\delta_i$ of a bus $i$, and let us denote the $p$ power state variables at time $k$ by a $p\times 1$ column vector, $x(k)$. Since we cannot directly measure the values of the power states, we have to rely on power flow meters to measure the readings of sensors installed in the power system. We further assume that there are $n$ power flow meters to provide sensor measurements at time $k$ denoted by an $n\times 1$ column vector, $y(k)$, each vector element corresponding to one sensor measurement. Thus, the dynamic state estimation problem is to estimate the state vector $\mathbf{x}(k)$ at time $k$ based on a sensor measurement vector  $\mathbf{y}(k)$, where $\mathbf{x}(k) \in \mathbf{R}^p$ is a vector containing $p$ state variables and
$\mathbf{y}(k) \in \mathbf{R}^n$ is a vector consisting of all meter measurements (i.e., active power injections, reactive power injections,
active power flows, and reactive power flows). The state estimation problem is considered dynamic since its power state variables and sensor measurements are varied with time $k$. Furthermore,
We assume that the system has the linear time invariant dynamics. That is, the state vector $\mathbf{x}(k+1)$ at time $k+1$ is related to the state vector $\mathbf{x}(k)$ at time $k$ as follows. 
$$\mathbf{x}(k+1)=\mathbf{Ax}(k)+\mathbf{w}(k),$$
where $\mathbf{A}=(a_{ij})$ is a $p \times p$ system matrix relating state vector at time $k$ to the state vector at time $k+1$, and $ \mathbf{w}(k)\sim \mathbf{N}(\mathbf{0}, \mathbf{\sigma_w^2I_p})$ is a process noise vector at the time $k$. 

The relationship between the state variables (voltage magnitude $V_i$ or angle $\delta_i$) and sensor measurements (real power injections $P_i$ and reactive power injections $Q_i$ at bus $i$) can be described by the following power flow equations: 
\begin{align*}
    P_i=V_i \sum_{j=1}^b V_j(G_{ij}\cos{(\delta_{i}-\delta_{j})}+B_{ij}\sin{(\delta_{i}-\delta_{j})})\\
     Q_i=V_i \sum_{j=1}^b V_j(G_{ij}\sin{(\delta_{i}-\delta_{j})}-B_{ij}\cos{(\delta_{i}-\delta_{j})})
\end{align*}  
where $i=1,2,..., b$ and $b$ is the total number of buses. Furthermore, $G_{ij}$ and $B_{ij}$ are the real and imaginary parts of the ($i, j$)-th element in the bus admittance matrix.

Moreover, we assume that the state vector is associated with the sensor measurement vector through the following relationship:
$$\mathbf{y}(k)=\mathbf{Cx}(k)+\mathbf{v}(k),$$
where $\mathbf{C}$ is an $n \times p$ matrix relating state vector to the measurement vector and is independent of time, and $ \mathbf{v}(k)\sim \mathbf{N}(\mathbf{0}, \mathbf{\sigma_v^2I_n})$ is the measurement noise vector at time $k$. The matrix $\mathbf{C}$ is determined by the topology of the power system. Furthermore, we assume that the attacker knows the matrix $\mathbf{C}$, and the observed measurement vector at time $k$ is
$$\mathbf{y}_o(k)=\mathbf{y}(k)+ \mathbf{\phi}(k),$$ 
where $\mathbf{\phi}(k) \in \mathbf{R}^n$ is the attack vector at time $k$. Assume that there are $m$ malicious sensor measurements and let $I=\{i_1, i_2, ..., i_m\}$ be the index set of the $m$ malicious sensor measurements. Then, the $i$th element of the attack vector $\mathbf{\phi}(k)$ is zero for all $i \notin I$.

\subsection{Measurement Consistency}\label{consistency}

In this section, we review the notion of a consistency check among sensor measurements used in \cite{b9}. The normal sensor measurements usually provide a good estimate of the true state variables, whereas the malicious sensor measurements are intended to provide a biased estimation of state variables, where such an estimation is quite different from the true value. Therefore, the consistency among the sensor measurements can be defined as follows. Note we assume $d>p$ in the definition since the state estimation is not unique otherwise.

\begin{definition}\label{d1}
Assume $d>p.$ Sensor measurements $y_{j_1}, ..., y_{j_d}$ are considered $\tau$-consistent if $$\min_x||\mathbf{C}(j_1,...,j_d)\mathbf{x}-(y_{j_1},...,y_{j_d})^T|| <\tau,$$ 
where $||*||$ is a $L_2$-norm and $\textbf{C}(j_1,...,j_d)$ is a matrix consisting of the $j_1$-th, $j_2$-th,..., and $j_d$-th rows of matrix $\mathbf{C}$. $\tau$ is the critical chi-square test statistics that satisfies
$$P(\textit{\textit{L }}<\tau)=\alpha,$$
where $L=||\mathbf{C}(j_1,...,j_d)\mathbf{x}-(y_{j_1},...,y_{j_d})^T ||$ follows a chi-square distribution with $d-p$ degree of freedom and $\alpha$ is significance level of the test.
\end{definition}

Several algorithms based on this consistency check have been introduced to detect malicious measurements \cite{b9}. A naive algorithm is to simply use the brute force approach to conduct the consistency check as follows.  

\subsubsection{Brute force approach} \label{brute}

The brute force approach examines all combinations of measurements in order to find the largest subset that is $\tau$-consistent. Starting with the set of all measurements $\textbf{y}=(y_1, y_2, ..., y_n)$, we estimate $\textbf{x}$ using the least square approach and check the $\tau$-consistency based on Definition 1. If it is $\tau$-consistent, we are done. Otherwise, we would test all subsets of $\textbf{y}$ with one fewer measurement. We would continue to do so until a $\tau$-consistent subset is obtained. (Note: the largest subset found by this brute force approach need not be unique since there could be multiple subsets of the same size that are $\tau$-consistent.)  

However, this approach is time-consuming in general. Therefore, Xiong and Ning \cite{b9} introduced a more computationally efficient probabilistic rank-based expanding approach. 

\subsubsection{Probabilistic Rank-Based Expanding Approach}\label{prob}

The brute force approach is inefficient when the number of sensors $m$ is large. A probabilistic rank-based expanding approach was initially proposed by Xiong and Ning in~\cite{b9}. Its main key is to use a probabilistic rank-based approach to find a collection of benign sensor measurements as a seed and then expand the seed to contain the rest of benign sensor measurements as many as possible for estimating a power state. The probabilistic rank-based expanding approach is less time-consuming, although there is no guarantee the obtained subset is the largest among all consistent subsets. It consists of the seeding phase and the expanding phase.

The seeding phase is to find a collection of benign sensor measurements as a seed. Since the majority of the sensors are benign, there is no need to examine all the subsets of size $p$ to find a good seed. Suppose there are at least

\begin{equation}
    \delta=n-\left\lfloor \frac{n-p-1}{2} \right\rfloor=p+\left\lfloor \frac{n-p+1}{2} \right\rfloor
\end{equation}
benign sensors. In this case, the probability that a randomly chosen subset of $p$ sensors is benign is
\begin{equation}
    P_{\delta}=\frac{\binom{\delta}{p}}{\binom{n}{p}}.
\end{equation}

Let $P_h$ be the probability that at least one subsets out of $h$ ($h>0$) randomly chosen subsets of $n$ measurements is benign. Then, $$P_h=1-(1-P_{\delta})^h.$$ However, since not all the subsets can be used to obtain an estimated state vector, $P_h$ is an approximation of the probability of finding a seed from $h$ randomly chosen subsets of $p$ measurements. 
 
To find a seed, we first calculate the number $h$ of subsets of size $p$ to ensure that $P_h$ is large enough, saying $0.995$. This can be done by solving the equation $$P_h=1-(1-P_{\delta})^h$$ and obtain 
 $$h=\left\lceil\frac{ \log{(1-P_h)} } { \log{ \left(1- \frac{ \binom{\delta}{p} }{\binom {n}{p}}  \right)}} \right\rceil$$
 Then, randomly choose $h$ subsets of size $p$ from $n$ measurements. For each chosen subset, solve for $x=(x_1, x_2,..., x_p)$ using the least square method. For each obtained solution $\hat{x}_j$, compute the absolute residuals: 
 $$r_{i, \hat{x}_j} =|\mathbf{C(i)x}-\mathbf{y}|,$$ 
 for all measurements $y_i$ for $i=1, 2, ..., n$. Since residual $r_{i, \hat{x}_j} $ reflects the consistency between $\hat{x}_j$ and $y_i$, the sorting score for each subset can be defined as the sum over the smallest $\delta$ squared residuals; i.e., $$ S_j=\sum_{i=1}^{\delta} (r_{i, \hat{x}_j})^2.$$ Intuitively, smaller $S_j$ indicates smaller overall residual and better estimate. Thus, we keep $n_{best}$,  where $0<n_{best}<h$, subsets that corresponds to the $n_{best}$ smallest sorting score used for the expanding phase. 
 
In the expanding phase, the objective is to expand the seed obtained in the seeding phase in order to find the largest consistent subset. For each seed obtained from the seeding phase, we sort all the measurements that are not in the seed $S$ according to their absolute residual. Start with the measurement with the smallest residual, check if it is $\tau$-consistent with the seed $S$. If yes, it is added to $S$. If not, it is discarded. Repeat this process until all measurements are checked. In section~\ref{method}, we extend the probabilistic rank-based approach to solve the SE problem when dynamic power systems are under FDI attacks. 

The consistency check is used to detect malicious sensor measurements in \cite{lin, smt, b9}. However, Liu et al. \cite{b10} proposed several FDI attacks that can bypass this $\tau$-consistency check. These FDI attacks are based on the idea of generating an attack vector $\mathbf{\phi}(k)$ such that $\mathbf{\phi}(k)$ is a linear combination of the column vectors of $\mathbf{C}$. That is, $\mathbf{\phi }(k)=\mathbf{Ce}(k) $, for some nonzero vector $\mathbf{e}(k)$. If $\mathbf{y}(k)$ at time $k$ can bypass the $\tau$-consistency check, then, as shown below, the malicious sensor measurement $\mathbf{y}_o(k)=\mathbf{y}(k)+\mathbf{\phi}(k)$ can also bypass the $\tau$-consistency check if $\mathbf{\phi }(k)=\mathbf{Ce}(k)$. 

Let $\mathbf{e}(k)=\hat{\mathbf{x}}_{o}(k)-\hat{\mathbf{x}}(k)$, where $\hat{\mathbf{x}}_{o}(k)$ is the estimated state vector obtained from $\mathbf{y}_o(k)$ at time $k$ and $\hat{\mathbf{x}}(k)$ is the estimated state vector obtained from $\mathbf{y}(k)$. Then, 
\begin{align*}
    || \mathbf{y}_o(k)-\mathbf{C}\hat{x}_{o}(k)|| &=||\mathbf{y}(k)+\mathbf{Ce}(k) -\mathbf{C}\hat{x}_{o}(k)||\\
    &=||\mathbf{y}(k)+ \mathbf{Ce}(k) -\mathbf{C}(\mathbf{\hat{x}}(k)+\mathbf{e}(k))||\\
    &=||\mathbf{y}(k)-\mathbf{C}\mathbf{\hat{x}}(k)||<\tau
\end{align*}
Therefore, if the attack vector $\mathbf{\phi}(k)$ is a linear combination of the column vectors of \textbf{C}, then the injected false data can bypass the detection. In the next subsection, we will consider this type of FDI attack, as well as a random attack. 

\subsection{The Attack Models}

Three attack models are considered in this paper. They are (1) the random attack considered by most researchers in the existing literature, (2) the specific sensor attack for a limited number of malicious sensors, and (3) the targeted attack that focuses on changing the state estimation for a particular set of state variables, i.e., specific voltage magnitude and/or angle.

\subsubsection{Attack Model 1}

The simplest one is a random attack model. This attack model assumes that an attacker can randomly select a subset of sensors and inject false data into them. To generate a random attack, an attacker randomly selects $m$ sensors and adds random errors to these sensor measurements. Formally, we let $I_m$ be a set of indices of $m$ sensors randomly selected from $n$ sensor readings. Then, $i$th component of the attack vector $\mathbf{\phi}(k)$ at time $k$ is 
\[\phi_i(k) =\begin{cases} 
      M*r & i\in I_m \\
      0 & i\notin I_m
   \end{cases},
\]
where $r \in N(0,1)$ and $M \in R$, and  
the malicious  sensor measurements $\mathbf{y}_o(k)$ at time $k$ is equal to $$\mathbf{\phi}(k)+\mathbf{y}(k),$$ 
where $\mathbf{y}(k)$ is an actual sensor measurement at time $k$. 

This attack model is the simplest in the sense that no much calculation is needed to determine the attack vector. The drawbacks are: (1) it is not targeted to specific state variables, and (2) it is assumed that the attacker can freely change measurements of the chosen sensors to any value at any time. 

\subsubsection{Attack Models 2 and 3}

The other two attack models are more intricate. They are called consistency attack models, which are based on the idea introduced in subsection \ref{consistency}. They are more complex than the random attack model and focus on injecting the false data that bypasses the consistency check. These attack models are developed based on Liu et al.'s work \cite{b10}. 

{\it a) Attack Model 2:} The specific sensor attack model assumes that the limited number of sensors are malicious and the attacker can change them at will. Suppose the attacker has access to $m$ sensors and let $I_m$
be an index set of these $m$ sensors. Let ${\phi}'$ be the reduced attack vector that is corresponding to $m$ malicious sensors. It is shown in \cite{b10} that if $${\phi}'=(I_{n}-B'^{-}B')d,$$ where the matrix $B'$ consists of the column vectors of $C(C^TC)^{-}C^T-\mathbf{I_n}$ that correspond to $m$ malicious sensors and is not full rank, $d$ is an arbitrary nonzero vector of length $m$, and
$\phi_i=\phi_i'$ 
if $i \in I_m$ and ${\phi}_i=0$ for $i \notin I_m$, then FDI attack can bypass the detection. See~\cite{b10} for the implementation of specific sensor attack.

{\it b) Attack Model 3:} In the targeted attack model, 
the attacker aims to change SE by the particular amount $\mathbf{c}(k)$ for chosen state variables $J_{target}=\{i_1, i_2, ..., i_u\},$ where $1\leq u<p$. That is, this attack satisfies the conditions that 

$$\hat{x}_{a,i}(k)=c_i(k)+\hat{x}_i(k) \mbox { for }i \in J_{target}$$
and
$$||y_o(k)-C\hat{x}_{a}(k)|| \leq \tau, $$
where $\hat{x}_{a,i}(k)$ is estimated $i$th state variable obtained by applying the weighted least square approach to malicious sensor measurements $\mathbf{y_o}$ at time $k$, denoted by $y_o(k)$; $c_i(k)$ is the specific error inject to $i$th state variable by an attacker at time $k$, and $\hat{x}_i(k)$ is estimated $i$th state variable before polluted by the attacker. See~\cite{b10} for the implementation of targeted attack.  

\section{The Proposed Methodology for Dynamic Power State Estimation}
\label{method}

In this section, we study the Mahalanobis transformation and its distance, propose robust approaches to detecting false data injection attacks in dynamic power state estimation, and give the time complexity analysis of the proposed approaches. Most importantly, we mathematically prove that the Mahalanobis distance is better than the Euclidean distance in terms of measuring an error in this study.

\subsection{Mahalanobis Transformation and Mahalanobis Distance}\label{mahalanobis}

All the existing studies on power state estimation use the Euclidean distance, which is actually based on the assumption that the measurement noise follows multivariate normal distributions with the variance-covariance matrix as an identity matrix given in subsection \ref{consistency}. However, in reality, the sensor measurements are correlated and the magnitude of the fluctuation among the sensor measurements are different. Therefore, the Euclidean distance cannot be directly applied to these correlated sensor measurements of different magnitudes. Hence, we use the Mahalanobis transformation to eliminate the correlation between the sensor measurements and to standardize the variance of each sensor measurement in order to satisfy the assumption that the measurement errors follow mutually independent standard normal distributions. 

Before defining the Mahalanobis transformation, we review some linear algebra concepts~\cite{everitt2011introduction, raykov2008introduction}. First, the eigenvalues and eigenvectors of a symmetric matrix are defined as follows.

\begin{definition}\label{eigen}
A $p \times p$ matrix $\mathbf{A}$ has eigenvalue $\lambda$ if there exists some p-dimensional vector $\gamma \neq 0$ for which $A\gamma =\lambda \gamma.$ The vector $\gamma$ is called the eigenvector corresponding to $\lambda$.
\end{definition}

The real symmetric matrix has real-valued eigenvalues.  If all eigenvalues of a matrix $\mathbf{A}$ are non-negative, then $\mathbf{A}$ is called positive semidefinite. If all eigenvalues of $\mathbf{A}$ are strictly positive, then it is called positive definite. Any symmetric matrix can be factored using the spectral decomposition (or Jordan decomposition) as shown in Theorem \ref{jordan} \cite{multivariate}.

\begin{theorem}\label{jordan}
Let $\mathbf{A}$ be a $p \times p$ symmetric matrix with the eigenvalues $\lambda_1, \lambda_2, ..., \lambda_p$ and the corresponding eigenvectors $\gamma_1, \gamma_2, ..., \gamma_p$. Then, $\mathbf{A}$ can be written as 
$$\mathbf{A}=U D U^T, $$
where $D =\mbox{diag}\{\lambda_1, \lambda_2, ..., \lambda_p\} $ and $U=(\gamma_1, \gamma_2, ..., \gamma_p).$
\end{theorem}

Using the spectral decomposition, we can define the power of a matrix as follows \cite{multivariate}. 

\begin{definition}
Let $\mathbf{A}$ be a $p \times p$ symmetric matrix such that $$\mathbf{A}=U D U^T, $$
where $D =\mbox{diag}\{\lambda_1, \lambda_2, ..., \lambda_p\} $ and $U=(\gamma_1, \gamma_2, ..., \gamma_p).$ Then for $a \in R$, we have $$\mathbf{A}^a=U D^a U^T, $$
where $D^a =\mbox{diag}\{\lambda_1^a, \lambda_2^a, ..., \lambda_p^a\}. $ 

In particular, if $a=-1$, then the inverse of $\mathbf{A}$ is $\mathbf{A}^{-1}=U D^{-1} U^T$. If $a=\frac{1}{2}$, then the square root of $\mathbf{A}$ is $\mathbf{A}^{\frac{1}{2}}=U D^{\frac{1}{2}} U^T$.
\end{definition}

\subsubsection{Mahalanobis Transformation}

In this subsection, we introduce the Mahalanobis transformation and show that this transformation standardizes a set of correlated variates $y$. The Mahalanobis transformation is defined as follows. 

\theoremstyle{definition}
\begin{definition}
Given an n-variate random vector $y$ with mean $E(y)=\mu $ and variance $Var(y)=\Sigma $.  Suppose that $\Sigma$ is positive definite. Then, the linear transformation
$$z=\Sigma^{-\frac{1}{2}}(y-\mu)$$
is called the Mahalanobis transformation. 
\end{definition}

In the following theorem, we show that the Mahalanobis transformation transforms $y$ to a standardized uncorrelated variate $z$. 

\begin{theorem}
Given n-variate random vector $y$ with mean $E(y)=\mu $ and variance $Var(y)=\Sigma $.  Suppose that $\Sigma$ is positive definite. Let $$z=\Sigma^{-\frac{1}{2}}(y-\mu).$$ Then, $E(z)=0$ and $Var(z)=I_p$.
\end{theorem}

\begin{proof}
\begin{align*}
E(z) &=E(\Sigma^{-\frac{1}{2}}(y-\mu)) \\
  &=\Sigma^{-\frac{1}{2}}E(y-\mu) \\
  &=\Sigma^{-\frac{1}{2}}(E(y)-E(\mu))\\
  & =\Sigma^{-\frac{1}{2}}(\mu-\mu)=0 
\end{align*}

and 
\begin{align*}
Var(z)&=Var(\Sigma^{-\frac{1}{2}}(y-\mu))\\
&=Var(\Sigma^{-\frac{1}{2}}y)\\
&=\Sigma^{-\frac{1}{2}}\Sigma\Sigma^{-\frac{1}{2}}=I &&\qedhere
\end{align*}

\end{proof}

Based on the central limit theorem, the noise follows a normal distribution under normal operations. Furthermore, after applying the Mahalanobis transformation, the noise has a zero mean and an identity covariance matrix based on Theorem 2. However, we do not know the true mean $\mu$ and variance $\Sigma$. Therefore, we use the sample mean $\bar{x}$ to estimate $\mu$ and the sample variance matrix 
$$S=\frac{1}{n-1}X^T(I_n-\frac{1}{n-1}1_n1_n^T)X$$ 
to estimate $\Sigma$. 

\subsubsection{Mahalanobis Distance}

After applying the Mahalanobis transformation, the variables are uncorrelated and standardized (theorem 3), and the familiar Euclidean distance can be used to define $\tau$-consistency. Alternatively, if we do not use the Mahalanobis transformation to decorrelate and standardize the variates, we can modify the definition of $\tau$-consistency use the idea of Mahalanobis distance. Usually, the Euclidean distance is used for measurements with the assumption that each element of $\mathbf{y}$ contributes equally to the calculation of the Euclidean distance. However, in reality, the magnitude of random fluctuations is different for each element of $\mathbf{y}$ and there are correlations between $\mathbf{y's}$. Hence, it is desirable to weigh each variable subject to the variance-covariance matrix. The Mahalanobis distance is defined as follows. 
\theoremstyle{definition}

\begin{definition}
Given a p-variate random vector $x$ with mean $E(x)=\mu $ and variance $Var(x)=\Sigma $. Suppose that $\Sigma$ is positive definite. 
Then the Mahalanobis distance of an observation $\textit{\textbf{y}}=(y_1, y_2, ..., y_n)$ from $\mu$ is 
$$D=\sqrt[]{(x-\mu)^T\Sigma^{-1}(x-\mu)}.$$
\end{definition}

The next theorem shows that the Euclidean distance is a special case of the Mahalanobis distance when the variance-covariance matrix is an identity matrix. That is, only when there is neither correlation between the variables nor the variation in fluctuations among the variables, the Mahalanobis distance is equivalent to the Euclidean distance; they are not equal to each other otherwise. In other words, the Mahalanobis distance provides a better measure of error since it does not assume that the variables are not correlated. 

\begin{theorem}
If $\Sigma=I$, then Mahalanobis distance is equivalent to the Euclidean distance.
\end{theorem}

A direct substitution of $\Sigma=I$ to the definition of Mahalanobis distance can prove the previous theorem. Using the idea of the Mahalanobis distance, we can redefine $\tau$ -consistency as follows.

\begin{definition}
Sensor measurements $Y=\{y_{j_1}, ..., y_{j_k}\}$ are considered $\tau$-consistent if 
$$\min_x\sqrt[]{(\mathbf{H}\mathbf{x}-Y)^T\Sigma^{-1} (\mathbf{H}\mathbf{x}-Y)}<\tau,$$ 
where $\mathbf{H}=C(j_1,...,j_k)$ is a matrix consisting of the $j_1$-th, $j_2$-th,..., and $j_k$-th rows of matrix $\mathbf{C}$.
\end{definition}

For the rest of the paper, we assume that either the Mahalanobis transformation has been applied to $y$ or the Mahalanobis distance is used for $\tau$-consistency checks. 

\subsection{Robust Approaches for Dynamic State Estimation}\label{SE approaches}

In this subsection, we present the proposed robust algorithms against FDI attacks for DSE. In subsection \ref{consistency}, the measurement consistency is defined based on the $L^2$-norm  $||\mathbf{C}(j_1,...,j_d)\mathbf{x}(k)-(y_{j_1},...,y_{j_d})^T||$ where $\textbf{C}(j_1,...,j_d)$ is a matrix consisting of the $j_1$-th, $j_2$-th,..., and $j_d$-th rows of matrix $\mathbf{C}$. Liu et al.~\cite{b10} showed the vulnerability of this type of consistency check and presented a new class of attacks that can bypass this consistency check. In this paper, we develop two alternative versions of this $\tau$ consistency check. These alternative versions do not rely on solving $\textbf{Cx}=\textbf{y}$ to find the MMSE estimate for $\textbf{x}$, so Liu et al.'s method~\cite{b10} for generating attack vector based on the column space of $\textbf{C}$ does not work in our proposed two versions of this $\tau$ consistency check. 

\subsubsection{Prediction based Consistency Approach (PCNA)}\label{method 1}

We estimate $\hat{x}_k$ by $\textbf{CA}\hat{x}_{k-1}$ and define the consistency check using the $L^2$-norm: 
\begin{equation}
    ||\mathbf{(CA)}(j_1,...,j_d)\mathbf{\hat{x}}(k-1)-(y_{j_1},...,y_{j_d})^T(k)||.
\end{equation}
Formally, it is defined as follows.

\begin{definition}\label{d7}

Assume $d>p.$ Sensor measurements $y_{j_1}, ..., y_{j_d}$ are considered $\tau$-consistent if $$\parallel\mathbf{(CA)}(j_1,...,j_d)\mathbf{\hat{x}}(k-1)-(y_{j_1},...,y_{j_d})^T(k)\parallel <\tau,$$ 
where $||*||$ is a $L_2$-norm and $\mathbf{(CA)}(j_1,...,j_d)$ is a matrix consisting of the $j_1$-th, $j_2$-th,..., and $j_d$-th rows of the product of matrix $\mathbf{C}$ and matrix $\mathbf{A}$. $\tau$ is the threshold that can be determined through a hypothesis test 
$$P(\textit{\textit{L }}<\tau)=\alpha,$$
where 
$$L=\parallel\mathbf{(CA)}(j_1,...,j_d)\mathbf{\hat{x}}(k-1)-(y_{j_1},...,y_{j_d})^T(k)\parallel$$ follows a chi-square distribution with $d-p$ degree of freedom and $\alpha$ is significance level of the test.

\end{definition}

This way of defining $\tau$-consistency is better than Definition \ref{d1} in two aspects. First, attack models proposed by Liu \emph{et al.} \cite{b11} do not work in this case. Secondly, it is computationally efficient by using $\mathbf{Ax}(k-1)$ to estimate $\mathbf{x}(k)$ instead of finding the estimate of $\mathbf{x}(k)$ using the minimum mean square error method (MMSE). 

Starting with all measurements: $$\mathbf{y}=(y_1, y_2, ..., y_n)^T(k),$$ we check if it is $\tau$ consistent. If it is not $\tau$ consistent, we remove such a  sensor, denoted as sensor $l$, that is corresponding to 
the largest residual $$r_l=\max_{1\leq i \leq  n}r_i,$$ 
where $r_i$ is the $i$-th component of $\mathbf{r}=\mathbf{CA}\mathbf{\hat{x}}(k-1)-\mathbf{y}^T(k)$.  Next, update the value of $\tau$ and check the $\tau$-consistency for the remain set of sensor measurements. Process repeat until a $\tau$ consistent set $I_c$ is obtained. 
See Algorithm~\ref{alg 4} for the implementation details.

\begin{algorithm}
\caption{Prediction based Consistency Approach (PCNA) Algorithm}\label{alg 4}
\begin{algorithmic}[1]
\Require{measurements $I$, previous state estimate }
\Ensure{a consistent subset of measurements }
\Statex 
\State{Compute the number $\delta$ of benign sensors required using equation 1}
\While{$|I| \geq \delta$ }   
\State{Compute the $L^2$-norm using measurement set $I$}
    \If {$L^2$-norm $<\tau$}
        \State{break}
    \Else
        \State {remove a measurement with largest residual}
        \State {from $I$ }
    \EndIf
\EndWhile
\end{algorithmic}
\end{algorithm}

\subsubsection{Combined Consistency and Kalman filter (CCKF) Approach} \label{method 2}

PCNA is time-efficient; however, it does not incorporate the current sensor measurements in SE. To added current sensor measurements information to SE, we need to define our $\tau$ consistency check differently. Before defining it, we introduce the Kalman filter estimation. The Kalman filter estimation is an optimal SE that combines the information from the system model, previous SE, and the sensors measurements. It is light on memory since the only previous state is needed to calculate the current state. These are the main reasons that the Kalman filter is selected for real-time DSE. The Kalman filter technique consists of two steps: prediction and update.

With the knowledge of the power system model, we calculate the predicted state estimate 
\begin{equation}\label{eqn:3}
    \hat{\mathbf{x}}(k|k-1)=\mathbf{A\hat{x}}(k-1|k-1)
\end{equation}
and its covariance matrix
\begin{equation}\label{eqn:4}
    \mathbf{P}(k|k-1)=A\mathbf{P}(k-1|k-1)A^{T}+\mathbf{\sigma_w^2I_p}.
\end{equation}

Use the current measurements, we can update the state estimate
\begin{equation}\label{eqn:5}
    \hat{\mathbf{x}}(k|k)=\hat{\mathbf{x}}(k|k-1)+\mathbf{K}(k)(\mathbf{y}(k)-\mathbf{C}\hat{\mathbf{x}}(k|k-1))
\end{equation}
and its updated estimate covariance
\begin{equation}\label{eqn:6}
\mathbf{P}(k|k)=(\mathbf{I}_p-\mathbf{K}(k)\mathbf{C})\mathbf{P}(k|k-1)
\end{equation}
where the Kalman gain
\begin{equation}\label{eqn:7}
\mathbf{K}(k)=\mathbf{P}(k|k-1)\mathbf{C}^{T}(\mathbf{CP}(k|k-1)\mathbf{C}^{T}+\mathbf{\sigma_v^2I_n}).
\end{equation}

The Kalman filter is an optimal SE that uses the information from both the system model and sensors measurements. However, if a subset of the sensor measurements is corrupt, it affects the SE. Hence, the consistency check is first applied to find a large subset of consistent sensor measurements, and then the Kalman filter is utilized to that subset for optimizing SE. 

To make the Kalman filter work under the FDI attack, we modify the Kalman filter approach by changing the Kalman gain $K$ to $K' \sim \mbox{Uniform} (K-r, K+r)$. where $r>0$ is some small number related to $K$ and $K-r >0$. This introduces some random variation to the Kalman gain to further reduce the attacker's ability to generating an effective FDI attack. Now we define the $\tau$-consistency based on this modified Kalman filter estimation as follows. 

\begin{definition}\label{d8}
Assume $d>p.$ Sensor measurements $y_{j_1}, ..., y_{j_d}$ are considered $\tau$-consistent if $$||\mathbf{C}(j_1,...,j_d)\mathbf{\hat{x}_{F}}-(y_{j_1},...,y_{j_d})^T|| <\tau,$$ 
where $||*||$ is a $L_2$-norm, $\mathbf{\hat{x}_{F}}$ is the modified Kalman filter estimate, and $\textbf{C}(j_1,...,j_d)$ is a matrix consisting of the $j_1$-th, $j_2$-th,..., and $j_d$-th rows of matrix $\mathbf{C}$. $\tau$ is the threshold that can be determined through a hypothesis test 
$$P(\textit{\textit{L }}<\tau)=\alpha,$$
where $L=||\mathbf{C}(j_1,...,j_d)\mathbf{\hat{x}_{F}}-(y_{j_1},...,y_{j_d})^T||$ follows a chi-square distribution with $d-p$ degree of freedom and $\alpha$ is significance level of the test.
\end{definition}

\begin{algorithm}[ht]
\caption{Probabilistic Rank-Based Expanding}\label{algorithm 1}
\begin{algorithmic}[1]
\Require{measurements, previous state vector, desired probability, and the number $n_{best}$ of best set want to obtain}
\Ensure{a consistent set of measurements}
\Statex 
\State{\textbf{Seeding Phase}}
\State{Compute the number $\delta$ of benign sensors required using equation 1 } 
\State {Compute the number $h$ of subsets needed to examine using equation 3}

\For{$j\gets 1$ to $h$}         
 	\State {Randomly select a subset $I_j$ of measurements}
    \State {Compute the state estimation based on $I_j$} 
    \State {Compute and sort the residuals}
    \State {Sum up the first $\delta$ smallest residual squares }
\EndFor
\State{Select $n_{best}$ subsets with smallest residual squares }
\State{\textbf{Expanding Phase}}
\For{each subset $S$ obtained from seeding phase}
	\For {each sensor measurement $y'$ not in $S$ }
	    \State{Check whether adding $y'$ to $S$ improves fitness}
	    \If {the fitness improves}
         \State{Add $y'$ to $S$}  
    \EndIf
	\EndFor
    \EndFor
\State{Keep the subset with the smallest residual sum of square}
\end{algorithmic}
\end{algorithm}

Before presenting our proposed robust approach for DSE in dynamic power systems, we first extend the probabilistic rank-based expanding approach to defend against FDI in dynamic power systems. Algorithm~\ref{algorithm 1} gives the pseudocode of this approach. Then, using Algorithm~\ref{algorithm 1}, Definition~\ref{d8}, and the Kalman filter described previously, we obtained the Combined Consistency with the Kalman filter (CCKF) Algorithm as follows.

\begin{algorithm}[ht]\label{algorithm 3}
\caption{Combined Consistency and Kalman filter (CCKF) Algorithm}
\begin{algorithmic}[1]
\Require{measurements, previous state vector estimation }
\Ensure{current state vector estimation }
\Statex 
\State{Find a large set of consistent measurements $y_c$ using either Algorithm \ref{alg 4} or  \ref{algorithm 1}}
\State {Compute the predicted state estimate by using equation~\ref{eqn:3}   }
\State {Compute the predicted covariance matrix by using equation~\ref{eqn:4}  }
\State {Compute the Kalman gain by using equation~\ref{eqn:7} }
\State {Compute the updated state estimate by using equation~\ref{eqn:5} and  $y_c$ obtained in step 1}
\State {Compute the updated covariance matrix by using equation~\ref{eqn:6}  }
\end{algorithmic}
\end{algorithm}

\subsection{Complexity Analysis}\label{compare}

Three consistency checks are introduced to find the consistent set of sensor measurements, assuming that the number of malicious sensors is less than half of the total number of sensors. In this section, we analyze the computational time complexity of them. 

There are several methods for implementing these consistency checks. 
The brute force approach introduced in section \ref{brute} is too time-consuming. If we do not know the actual number of the attacked sensors, in the worst case scenario, there are $$\left( \begin {array}{c} n\\ n \end{array} \right)+...+\left( \begin {array}{c} n\\ \lceil \frac{n}{2} \rceil\end{array} \right) = 
  \left \{ \begin {array} {ll}
\frac{ 2^n- \left ( \begin{array}{c} n\\ \frac{n}{2} \end{array} \right)} {2} & \mbox{if n is even} \\  \\
2^{n-1} & \mbox{if n is odd} \end{array} \right. $$
MMSE operations. Hence, the time complexity is $\mathcal{O}(2^{n})$ in the worst case scenario. 

In contrast, the probabilistic rank-based expanding approach can find a large subset of measurements for SE while being efficient. The seeding phase primarily consists of finding MMSE and calculating the residuals. As stated in section IV, we only need to examine 
\begin{equation}
    h=\left\lceil\frac{ \log{(1-P_h)} } { \log{ \left(1- \frac{ \binom{\delta}{p} }{\binom {n}{p}}  \right)}} \right\rceil
\end{equation}
subsets of $p$ measurements to find a seed for expanding phase. Hence, the time complexity for solving the linear equation is $\mathcal{O}(hp^{2.376})<\mathcal{O}(\binom{n}{k} p^{2.376})$ whereas the time complexity for calculate the measurement residual in the seeding phase is $\mathcal{O}(nph)<\mathcal{O}(\binom{n}{k} np)$. Therefore, the time complexity of seeding phase is $\mathcal{O}(\binom{n}{k} (np+p^{2.376}))$. 
For the expanding phase, we only need to perform $n-p$ Kalman filter estimation operations. Thus, the most time-consuming step is from the seeding phase with time complexity $\mathcal{O}(\binom{n}{k} (np+p^{2.376}))$, as the time complexity of the Kalman filter estimation is the same as the time complexity for solving a linear system of equations. 

However, there is no need to do the rank-based expanding approach if Definition 7 is used. In this case, the dominant operations consist of matrix multiplications and the calculation of the measurement residual. The best known algorithm for multiplying an $n \times p$ and a $p \times p$ matrix runs as $\mathcal{O}(np^{2})$, whereas the time complexity for calculating the measurement residual is $\mathcal{O}(np)$. Therefore, the time complexity is $\mathcal{O}(np^{2})$, and the consistency check based on Definition 7 is the most time-efficient. 
\subsection{Discussion}
In this subsection, we would like to discuss the two potential extension directions of our proposed approaches as follows.

As shown above, we have developed robust approaches against false data injection attacks, where it is assumed that power sensor data are collected and sent to a control center for dynamic power state estimation. In a large-scale power grid system, however, it is not only very expensive but also infeasible to send all the sensor data collected in different locations to a single central location such as the control center. This is because real-time monitoring a power system must result in large sensor data and such a big data transfer over networks is very challenging in terms of performance and security. Thus, it is better to achieve power state estimation in a distributed computing fashion. That is, power sensor data will be processed in a local control center that is close to where those data are collected. Such a distributed computing framework can dramatically reduce the sizes of data that need to be transferred and the difficulty of the transfer of those data, resulting in the performance improvement of data transfer. Of course, such a framework may arise other research challenges that are usually seen in a distributed computing. However, a lot of studies have been conducted to addressed such challenges. For example, 
we can propose to implement a blockchain-based orchestrator to automate the management, coordination, and organization of the proposed approaches in the local control centers. The orchestration supports the efficient delivery of distributed computing resources for the local state estimation under false data injection attacks, whereas blockchain-based implementation eliminates trust issues when local control centers may belong to different entities. More precisely, we may follow the framework proposed in~\cite{qu2020decentralized} to implement a blockchain-based orchestrator. In this case, each local state estimation update is uploaded to its associated miner. This local state estimation is verified and shared among all other miners, and then the global state estimation is achieved through the collaborative computation among these miners.

Furthermore, distributed computing encounters a variety of security attacks. For example, a new type of data poisoning attacks has been recently developed to attack federated learning systems, a new type of distributed computing systems~\cite{zhang2019poisoning}. The main idea of this attack is to use a generative adversarial network (GAN) to artificially generate data similar to real data and then to update the weights of each local neural network model based on these artificial data instead of real data. This attack is stealthy because the generated data are apparently indistinguishable from real data. Moreover, in order to increase the influence of the updates of local model weights on the global neural network model, the attacker scales up their local model weight updates by some large values, i.e., 20–100 as proposed in~\cite{zhang2019poisoning}. Though this scale-up increases the attack success rate, the updated weights provided by an attacker is very different from the weight updates of local neural networks obtained through normal operations without any attacks. In our further study, we may apply our proposed consistency check to local model weights for defending against such data poisoning attacks. More precisely, we may follow the consistency check framework proposed in~\cite{lin2021active}.

\section{Evaluation}\label{evaluation}
We empirically evaluate the proposed approaches through experiments using an IEEE 14 bus system in this section. More specifically, we give our evaluation setup, present our parameter selection of three attack models, and conduct the performance comparison of the proposed algorithms with existing ones in terms of accuracy, runtime, and the effect of injection error on estimation.

\subsection{Evaluation Setup}

The IEEE 14 bus power system is used to compare three definitions of consistency and to evaluate all four approaches: the least square approach, the Imhotep-SMT approach, the CCKF approach, and the MEE-UKF approach~\cite{10}. Since the Imhotep-SMT approach is evaluated using the 14 bus system in~\cite{smt}, we compare the Imhotep-SMT approach with the CCKF approach using the same bus system. This 14 bus power network consists of five synchronous generators and 14 buses \cite{2}. To ensure that SE is feasible, sensor number 35 is assumed to be safe as well. That is, the attacker does not have access to sensor number 35. The matrices $A, B,$ and $C$ modeling the power network are derived in~\cite{2}. The sensor readings at time $k$ is generated by calculating the sensor measurements using $\mathbf{Cx}(k)+\mathbf{v}(k)+\mathbf{\phi}(k)$, where $\mathbf{v}(k)$ follows the normal distribution with mean 0 and variance $10^{-1}I_n$, and $\mathbf{\phi}(k)$ is the attack vector given as before. 

\subsection{Parameter Selection of Three Attack Models}

We have three attack models. The simplest attack model is a random attack model (model 1). For this attack model, the attacker randomly selects $m=14$ sensors and introduces random errors to these sensor measurements. The next two attack models (consistency attack models) are more intricate. The consistency attack models are based on the idea introduced in subsection \ref{consistency}. They are more complex than model 1 and focus on introducing the attack that bypasses the consistency check. Attack models 2 and 3 are developed based on Liu et al.'s work \cite{b10}. For attack model 2, we assume the attacker has access to 14 meters; i.e., meters numbering 1-5, 13, 15-20, 28, and 33 are malicious. The attack vector for Attack Model 2 is ${\phi}(k)=(I_{35}-B'^{-}B')d$, where the matrix $B'$ is the column vectors of $C(C^TC)^{-}C^T$ that are corresponding to meters 1-5, 13, 15-20, 28, and 33 and $d=50*1_{14}$ (In general, $d$ is an arbitrary nonzero vector). Attack Model 3 is called the targeted attack Model. The targeted attack model introduces specific errors to chosen target state variables using Algorithm 2. 

Before the simulation, we also specified some parameters for these approaches. For the Imhotep-SMT approach, we specified the upper bound for the maximum number of the attacked sensors to be 14. For the CCKF approach, we set $\alpha=0.5\%$ and $P_{h}=0.995$ since we want to have a high confidence, i.e., the 99.5\% confidence, to ensure that we are able to find a large consistent set whose type 1 error is $\alpha=0.5\%$ within first $h$ random subsets tested. We set the process covariance matrix as $10^{-7}I_p$ and the noise covariance matrix as $10^{-1}I_n$ to see the performance of our approach under a relatively large measurement uncertainty comparing to process uncertainty~\cite{8309126}. 

\begin{figure}[ht]
\centering
\includegraphics[width=\columnwidth]{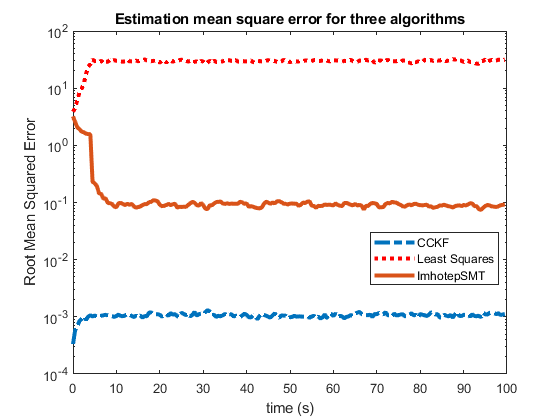}
\caption{Root mean squared error in the power system under attack model 1}
\label{fig 1}
\end{figure}

\begin{figure}[ht]
\centering
\includegraphics[width=80mm]{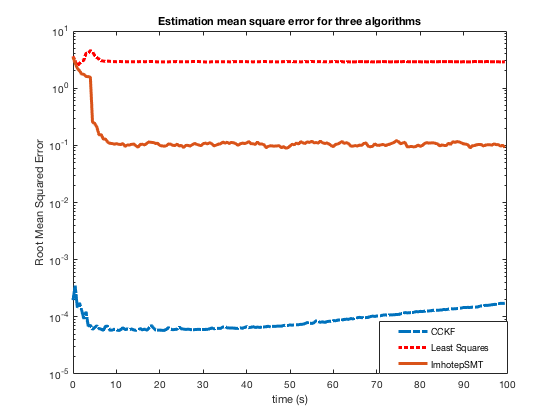}
\caption{Root mean squared error in the power system under attack model 2}
\label{fig 2}
\end{figure}

\begin{figure}[ht]
\centering
\includegraphics[width=80mm]{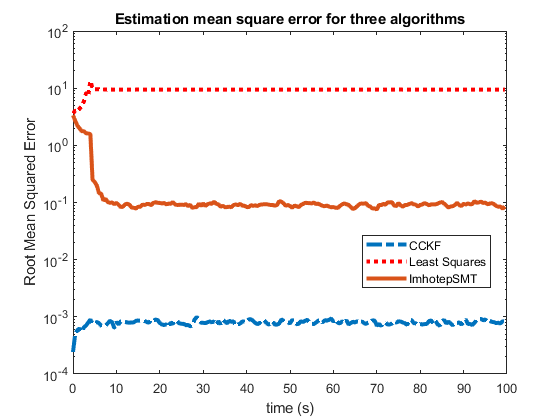}
\caption{Root mean squared error in the power system under attack model 3}
\label{fig 3}
\end{figure}

The simulation is repeated 100 times, and the root-mean-squared error (RMSE) of attack-resilient SE is reported based on the result of these 100 simulations. Figures~\ref{fig 1}, \ref{fig 2}, and \ref{fig 3} show the root mean squared error introduced by attack models 1, 2, and 3, respectively. As shown, the least square method performs poorly under all three attack scenarios. The Imhotep-SMT approach performed better with RMSE about 0.1, and the CCKF approach is best with the smallest RMSE. After approximately 10 seconds of simulation time, the errors are greatly reduced when either the CCKF approach or the Imhotep-SMT approach is used.  

Furthermore, we compare our approach with the MEE-UKF approach~\cite{10}. The comparison setup is the same as~\cite{10}, except we use the root mean square error (RMSE) instead of the mean absolute error (MAE) as a measure of performance. Both approaches utilize the Kalman filter, although our approach also uses the Kalman filter's prediction step for consistency check. As shown in Figure~\ref{fig 8}, our approach has a somewhat lower root mean square error and provides a more stable result since it first filters out the inconsistent sensors before the state estimation. 

\begin{figure}[ht]
\centering
\includegraphics[width=80mm]{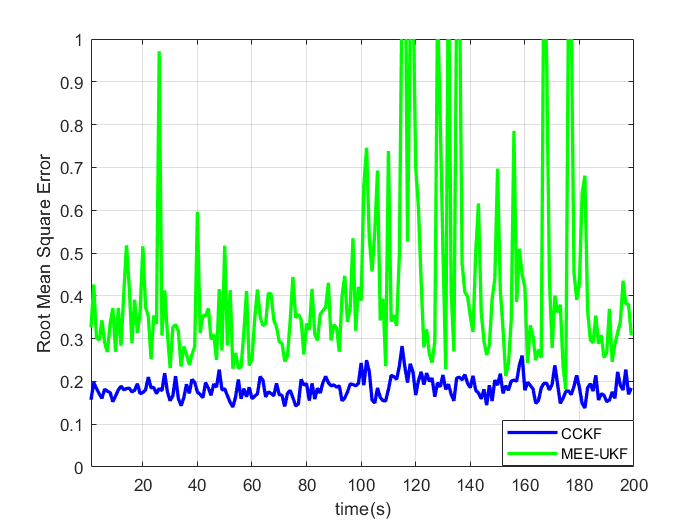}
\caption{Model Comparison: CCKF vs. MEE-UKF}
\label{fig 8}
\end{figure}

\subsection{Performance Comparison}

Since the Kalman filter uses both information from the sensor measurements and the system model, the Kalman filter gives a more precise state estimation than the least square approach. Combining the consistency check with the Kalman filter gives a new approach for SE against FDI attacks. By the observation matrix definition, the measurement matrix $C$ is a submatrix of $O$. Since the Imhotep-SMT approach also considers the consistency check but in terms of the observation matrix, the Imhotep-SMT approach gives a more precise estimation. However, the Imhotep-SMT approach assumes that the bound on the maximum number of malicious sensors and the bound on the sensor noise are known in addition to knowing that the number of malicious sensors is sparse. For the CCKF approach, we assume that the noise is Gaussian and the number of malicious sensors is less than half of the total number of sensors. With the Gaussian noise assumption, if the consistency check can remove all malicious sensors, the Kalman filter estimation is optimal. In this subsection, the performance of two proposed consistency checks based on Definitions 7 and 8 are evaluated and compared with the existing consistency check defined in Definition 1. 
Figures~\ref{fig 4} and~\ref{fig 5} show the RMSE in SE of the power system under the random attack and the consistency attack, respectively. 
In the first experiment (Figure~\ref{fig 4}), we generate a random attack based on Algorithm 1. For this attack model, the attacker randomly selects $m=14$ sensors and introduces random errors to these sensor measurements. As shown in Figure~\ref{fig 4}, the performances of the three consistency checks are about the same. In the second experiment (Figure~\ref{fig 5}), we generate a specific meter attack assuming that the attacker has access to meters 1-5, 13, 15-20, 28, and 33. As shown in Figure~\ref{fig 5}, the $\tau$-consistency check based on Definition 7 performs best among the three.

\begin{figure}[ht]
\centering
\includegraphics[width=80mm]{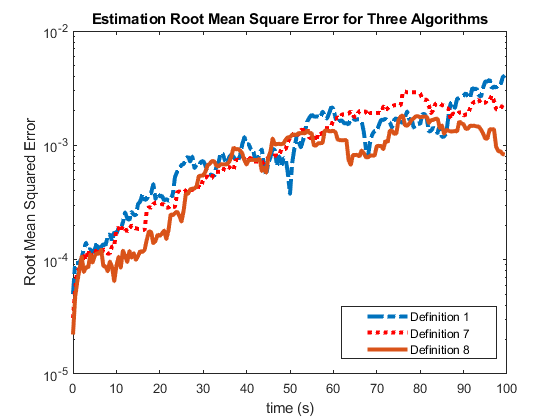}
\caption{Three consistency checks based on Definition 1, 7, and 8 have similar performances under random attack.}
\label{fig 4}
\end{figure}

\begin{figure}[ht]
\centering
\includegraphics[width=80mm]{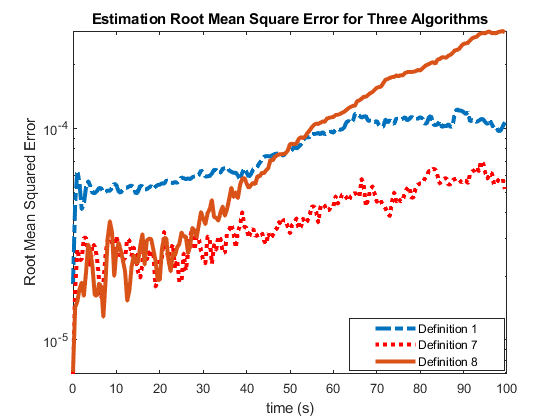}
\caption{The consistency check based on Definition 7 outperforms others under the Consistency attack.}
\label{fig 5}
\end{figure}

\subsection{Runtime Comparison }

Table~\ref{table 1} shows the execution times for the SE of the power system under the random attack. Without loss of generality, we increase the number of state variables and the number of sensors simultaneously. We set $n=3p$ as in \cite{smt} for comparison purposes. As $p$ increases, the average runtime also increases. When $p=10,$ the average runtime for the CCKF approach based on Definition 8 is the smallest. However, the average runtime for it increases faster than that of PCNA. When $p=50,$ the average runtime for the PCNA becomes the smallest. When $p=150,$ the average runtime for the MMSE is triple the average runtime for PCNA, whereas the average runtime for CCKF is more than ten times the average runtime for the CCKF. This shows that PCNA based on Definition 7 is the most time effective and confirms the runtime analysis of section 3. 

\begin{table}[ht]
\caption{Runtime Comparison}
\begin{center}
\begin{tabular}{|l|l|l|l|l|l|l|}
\cline{1-7}
 & \multicolumn{2}{l|}{MMSE} & \multicolumn{2}{l|}{PCNA}  &   \multicolumn{2}{|l|}{CCKF}       \\ \cline{1-7}
p& Mean      & Sd         & Mean & Sd          & Mean & Sd  \\ \cline{1-7}
10 & 0.844   &  0.101    &0.396  &0.029  &0.317 &0.034 \\ \cline{1-7}
25  & 1.176  & 0.148     &0.345  &0.004  &0.467 &0.062 \\ \cline{1-7}
50  & 1.497  & 0.190     &0.710  & 0.031 &3.628  & 0.544\\ \cline{1-7}
75  & 1.997  & 0.227     & 0.833 & 0.024 & 9.69 & 1.337\\ \cline{1-7}
100  & 2.694 &  0.379    & 0.993 & 0.027 & 14.909 & 1.341\\ \cline{1-7}
125  & 4.246 & 0.583     &  1.163& 0.031 & 16.123 &1.841 \\ \cline{1-7}
150  & 5.912 &  0.878    &  1.439& 0.022 & 19.312 & 2.320\\ \cline{1-7}

\end{tabular}
\end{center}
\label{table 1}
\end{table}

\subsection{Effect of injection error on estimation}

To see the effect of the injection error on SE, we double the injection errors and rerun the experiment in subsection C. As shown in Figures~\ref{fig 6} and~\ref{fig 7}, the performance of CCKF (based on Definition 8) and PCNA (based on Definition 7) are similar, whereas their performances are significantly better than the performance of MMSE (based on Definition 1) under both consistency attack and random attack. 

\begin{figure}[ht]
\centering
\includegraphics[width=80mm]{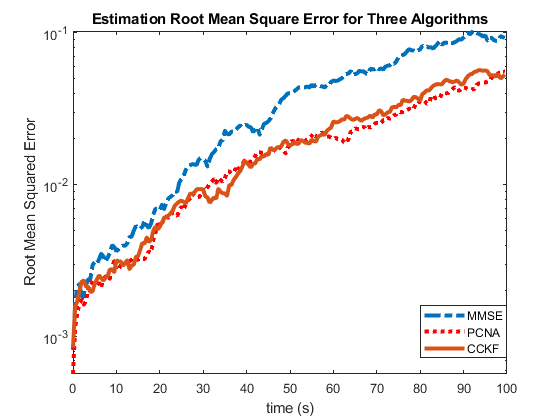}
\caption{Root-Mean-Square Error of SE under a random attack}
\label{fig 6}
\end{figure}

\begin{figure}[ht]
\centering
\includegraphics[width=80mm]{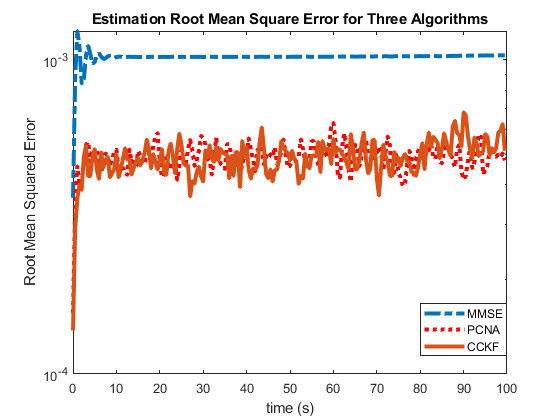}
\caption{Root-Mean-Square Error of SE under a consistency attack}
\label{fig 7}
\end{figure}

\section{Related Work}\label{related}
In this section, we review some related work to our study. Many researchers have considered the problem of detecting and identifying FDI attacks \cite{coutinho2009anomaly}, \cite{fawzi2014secure}, \cite{b1}-\cite{b26}, but they mainly focus on static power state estimation instead of a dynamic one studied in this research. For instance, Pires et al. \cite{b11} and Vedik and Chandel \cite{b13} proposed robust iteratively reweighted least squares (IRLS) approach and least winsorized square (LWS) approach for SE, respectively. The least winsorized square is robust in that it eliminates an outlier before making the SE. Our proposed methodology also performs the SE by first removing the outlier or potential malicious sensors.
Teixeira et al. \cite{11} proposed the false data attack under both linear and nonlinear estimators, assuming that the attacker only possesses a perturbed model. Xiong and Ning \cite{b9} also studied the FDI attacks against power system SSE and counter measurement. Nevertheless, Yao Liu et al. \cite{b10} showed the vulnerability of this SSE; i.e., \newHL{an} FDI attack can be designed that introduces arbitrary errors that bypass detection even when the attacker has limited resources. 

Different from the SSE, the DSE is more robust in the sense that it includes the information from the previous state and system model \cite{13}. A number of techniques for DSE in power systems have been developed \cite{10}, \cite{b17}, \cite{9}, \cite{2020}, \cite{valverde2011unscented}-\cite{wu}, \cite{b14}, \cite{u2}, \cite{b19}. For instance, an extended particle filter (PF) is used to estimating the dynamic states of a synchronous machine \cite{14}, and Zhang et al. \cite{u2} developed an adaptive mixed Kalman/$H_{\infty}$ filtering for the distribution network. Mandal et al. \cite{b17} proposed two algorithms that incorporated the measurement function nonlinearities in the extended Kalman filter (EKF) scheme for DSE. Pang et al. \cite{2020} proposed a Kalman filter based output tracking control system against FDI \newHL{attacks}. Kundu et al. presented an anomaly detector based on an auto-encoder \cite{b37}.

Yang et al. \cite{b14} developed the following five novel attack approaches: maximum magnitude-based attack, wave-based attack, positive deviation attack, negative deviation attack, and mixed attack that can bypass the anomaly detection. These approaches are based on the idea of injecting a small error that is within the specified tolerance. In addition, they developed temporal-based detection algorithms to defend against those attacks. They assumed that the measurement errors before the attack follow a distribution, whereas the measurement errors after the attack follow a different distribution. However, in reality, an attacker could inject the targeted specified error that constantly moves the system through multiple transient phases. Shoukry and Tabuada \cite{b27} designed two computationally efficient algorithms for SE using event-triggered techniques. However, the work does not consider the sensor and process noise and only considered the random attack model. On contrary, this paper considers both sensor and process noise. In addition, our proposed algorithm is evaluated using both the random attack model and consistency attack models described in the previous section. Wu et al.~\cite{wu} proposed a novel sliding-mode observer for SE and an event-triggered scheme for saving limited computational resources. Lyu et al. \cite{lyu} considered a different aspect of the SE, the transmission reliability for SE. In \cite{smt}, Tabuada introduced the Imhotep-SMT approach for secure SE. The technique consists of two main steps: detection and estimation. The idea behind the detection of the malicious sensors is similar to the consistency check introduced in the previous section. However, instead of using the norm $||\mathbf{y}(k)-\mathbf{C}\mathbf{x}(k)||$, it uses $||\mathbf{Y}(k)-\mathbf{O}\mathbf{x}(k)||$ where 

\begin{equation*}
\mathbf{Y}(k)=
\begin{bmatrix}
    \mathbf{y}(k-\eta+1)\\
    \mathbf{y}(k-\eta) \\
    \vdots \\
    \mathbf{y}(k)
\end{bmatrix}
\end{equation*}
is the measurement matrix, and $\eta \leq p$ is an integer selected that \newHL{guarantees} the system observability matrix  
\begin{equation*}
\mathbf{O}=
\begin{bmatrix}
    \mathbf{C}\\
    \mathbf{CA}\\
    \vdots \\
    \mathbf{CA}^{\eta-1}
\end{bmatrix}
\end{equation*} has full rank. Let  $I$ be the subset of sensors that passes detection. Then, the estimated state $\mathbf{x}(k)$ is the solution to the equation $\mathbf{Y}_I(k)=\mathbf{O}_I\mathbf{x}(k)$. This approach has the same limitation as the general MMSE as shown in Theorem~\ref{theorem 4}. 

\begin{theorem}\label{theorem 4}
Let $\mathbf{y}(k)$ be the original sensor measurement vector at time $k$ that can pass the bad measurement detection of the Imhotep-SMT approach. Then, the received malicious measurement vector $$\mathbf{y}_a(k)=\mathbf{y}(k)+\mathbf{Ce}$$ at time $k$  can pass the bad measurement detection if  $$\mathbf{e}=\mathbf{O}^{-1}(\mathbf{1}_{\eta}	\otimes\mathbf{\phi}),$$ 
where $\mathbf{\phi}$ is the nonzero attack vector injected by an attacker. 
\end{theorem}

\begin{proof}
Since $\mathbf{O}$ has full rank, then $\mathbf{O}$ has a left inverse $\mathbf{O}^{-1}$ such that $$\mathbf{O}^{-1}\mathbf{O} = I_p.$$ Therefore, $$\mathbf{e}=\mathbf{O}^{-1}(\mathbf{1}_{\eta}	\otimes\mathbf{\phi})$$ exists.
Let $\hat{\mathbf{x}}(k)$ be the estimated state vector obtained from the $\mathbf{y}(k)$, and $\hat{\mathbf{x}}_a(k)$ be the estimated state vector obtained from the $\mathbf{y}_a(k)$. Then, $\mathbf{Y}(k)$ and $\mathbf{Y}_a(k)$ are a corresponding matrix generated from $\mathbf{y}(k)$ and $\mathbf{y}_a(k)$, respectively. And we have that

$$|| \mathbf{Y}_a(k)-\mathbf{O}\hat{x}_{a}(k)|| =||\mathbf{Y}(k)+ (\mathbf{1}_{\eta} \otimes \mathbf{\phi} )-\mathbf{O}(\mathbf{x}(k)+\mathbf{e})|| $$
  \begin{align*}
&=||\mathbf{Y}(k)+ (\mathbf{1}_{\eta} \otimes \mathbf{\phi} )-\mathbf{O}(\mathbf{x}(k)+\mathbf{O}^{-1}(\mathbf{1}_{\eta}	\otimes\mathbf{\phi}))|| \\
&=||\mathbf{Y}(k)+ (\mathbf{1}_{\eta} \otimes \mathbf{\phi} )-\mathbf{O}\mathbf{x}(k)-(\mathbf{1}_{\eta} \otimes \mathbf{\phi} )|| \\
&=|| \mathbf{Y}(k)-\mathbf{Ox}(k)||
\end{align*}  

\noindent Therefore, if $|| \mathbf{Y}(k)-\mathbf{Ox}(k)||<\tau$, $|| \mathbf{Y}_a(k)-\mathbf{O}\hat{x}_{a}(k)||<\tau$. 
\end{proof}
This theorem shows that the constructed attack vector $\mathbf{\phi}$ can easily bypass the detection. On the contrary, our proposed consistency checks overcome this limitation. 

\section{Conclusion and Future Work}
\label{conclusion}

A power grid is a typical energy-based Cyber Physical System (CPS) that is vital to our daily life, but it, at the same time, is susceptible to various cyber attacks. A successful attack on power systems not only results in a significant economic loss, but it may also cause loss of human life. State estimation is needed for controlling and monitoring the state of power systems based on readings of sensors placed at important power grid components. Existing studies focused on static power state estimation. In this paper, instead, we investigated power state estimation, where PCNA and CCKF approaches were proposed and implemented to estimate the dynamic states of a power system. The two approaches are robust. The experimental studies illustrated that the performances of PCNA and CCKF are similar, even though PCNA is more time-efficient. The performance of the CCKF approach was also compared with the two well-known approaches, the Imhotep-SMT approach, and the least square approach, under three different attack models: a random attack model and two consistency attack models, which gave typical FDI attacks in a power system. Our experimental results demonstrated that the proposed approach outperforms both the Imhotep-SMT approach by two orders of magnitude and the least square approach by four orders of magnitude.  Furthermore, we compared our approach with the MEE-UKF approach to show that our approach provides a more stable result, though both approaches have similar performance. Moreover, all the existing studies used the Euclidean distance in power state estimation. In this research, the little-known but useful Mahalanobis distance was presented and used for $\tau$-consistency calculation. Finally, we investigated the properties of the Mahalanobis transformation and the Mahalanobis distance through theoretical analysis to show that the Mahalanobis distance is a better measure of the error than the Euclidean distance does.

In the future work, we will investigate the performance of the proposed approaches on real-world power grid data. We will also consider extending the proposed approaches by using a blockchain-based orchestrator and studying data poisoning attacks. 

\section*{Acknowledgment}

We acknowledge National Science Foundation to partially sponsor Dr. Kaiqi Xiong\textquotesingle s work under grants CNS 1620862 and 1620871, and BBN/GPO project 1936 through NSF/CNS grant. The views and conclusions contained herein are those of the authors and should not be interpreted as necessarily representing the official policies, either expressed or implied of NSF.

\bibliographystyle{cas-model2-names}
\bibliography{reference.bib}
\vskip6pt

\end{document}